\documentclass[12pt]{article}

\usepackage{amsmath} 
\usepackage{amsfonts} 
\usepackage{microtype} 
\usepackage{hyperref} 
\usepackage{amsthm} 
\usepackage[shortlabels]{enumitem} 
\usepackage{amssymb} 
\usepackage{comment}
\usepackage{tikz} 
\usepackage{subcaption}
\usepackage{authblk}
\usepackage{geometry}
\newtheorem{thm}{Theorem}
\numberwithin{equation}{section}

\newtheorem{defn}[thm]{Definition}

\newtheorem{cor}[thm]{Corollary}
\newtheorem{prop}[thm]{Proposition}
\newtheorem{lemma}[thm]{Lemma}

\newcommand{\col}{\text{col}}

\newcommand{\phase}{\text{phase}}

\newcommand{\RR}{\mathbb{R} }

\newcommand{\II}{I}
\newcommand{\CC}{\mathbb{C}}
\newcommand{\MM}{\mathcal{M}}

\newcommand{\A}{\mathcal{A}}
\newcommand{\B}{\mathcal{B}}

\newcommand{\fR}{f_{\text{R}}}

\begin{document}
\title{A note on Douglas-Rachford, gradients, and phase retrieval}
\author[1]{Eitan Levin}
\author[2]{Tamir Bendory}
\affil[1]{The Program in Applied and Computational Mathematics, Princeton University, Princeton, NJ, USA}
\affil[2]{School of Electrical Engineering, Tel Aviv University, Tel Aviv, Israel}
\maketitle

\begin{abstract}
The properties of gradient techniques for the phase retrieval problem have received a considerable attention in recent years.  
In almost all applications, however, the phase retrieval problem is solved using a family of algorithms that can be interpreted as variants of Douglas-Rachford splitting. 
In this work, we establish a connection between Douglas-Rachford and gradient algorithms.
Specifically, we show that in some cases a generalization of Douglas-Rachford, called relaxed-reflect-reflect (RRR), can be viewed as gradient descent on a certain objective function.
The solutions coincide with the critical points of that objective, which---in contrast to standard gradient techniques---are not its minimizers.
Using the objective function, we give simple proofs of some basic properties of the RRR algorithm. Specifically, we describe its set of solutions, show a local convexity around any solution, and derive stability guarantees. Nevertheless, in its present state, the analysis does not elucidate the remarkable empirical performance of RRR and its global properties. 
\end{abstract}

\section{Introduction}

For a given sensing matrix  $A\in\CC^{m\times n}$ and magnitudes $b\in\RR^{m}_{\geq 0}$, the goal of the phase retrieval problem is solving the system of $m$ equations
\begin{equation} \label{eq:pr}
|Ax_0| = b,
\end{equation}
where the absolute value is taken entry-wise.
The matrix $A$ usually represents a Fourier-type transform, such as the DFT matrix $(m=n)$ or the over-sampled DFT matrix $(m> n)$.
Phase retrieval can be formulated conveniently as a feasibility problem:   finding   a point $x\in \A\cap \B$, where $\B$ is the set of all signals that satisfy~\eqref{eq:pr}, namely, 
\begin{equation} \label{eq:setB}
\mathcal{B}=\{y\in\CC^m:\ |y|=b\},
\end{equation}
and the set $\A$ encodes application-specific additional knowledge about the solution, such as sparsity or known support.
The phase retrieval problem and some of its applications are discussed in Section~\ref{sec:phase_retrieval}.

In the last decade, the computational and theoretical aspects of the phase retrieval problem have received much attention. 
To facilitate  the mathematical analysis, it became fashionable to investigate a toy model---one that does not appear in applications---where the entries of $A$ are drawn i.i.d.\ from a   normal distribution (or similar statistical models); 
hereafter we refer to the problem of recovering a signal from such measurements as the \emph{random phase retrieval problem}. 
The most  popular algorithms for this problem are based on minimizing different non-convex loss functions (e.g., non-convex least squares) using first-order gradient techniques; see for instance~\cite{candes2015phase_WF,chen2017solving,wang2017solving,cai2016optimal}.
Notably, this line of papers derived solid theoretical guarantees by showing that the non-convexity of the problem is usually benign when $m$ is sufficiently larger than $n$~\cite{sun2018geometric,Chen2019}.

Unfortunately, it is now clear that the  random phase retrieval  problem is considerably easier than the actual phase retrieval problem, when $A$ is a Fourier-type matrix.  
For most phase retrieval applications, the algorithms proposed for the random phase retrieval setup fail: the non-convexity is not benign and gradient-based algorithms are trapped in local minima, far from a global solution; see an elaborated discussion in~\cite{Elser2017}.
Consequently, this substantial body of literature have had  only a minor effect  on practical applications.
Instead, many heuristic algorithms are used in practice, including the hybrid input-output (HIO)~\cite{Fienup1982}, difference map~\cite{elser2003phase}, relaxed averaged alternating reflections (RAAR)~\cite{Luke2005}, and relaxed reflect reflect (RRR)~\cite{Elser2017a}.  
All these algorithms can be understood as generalizations of the Douglas-Rachford algorithm~\cite{douglas1956numerical}; see Section~\ref{sec:DR} for an introduction.
By a slight abuse of terminology, we shall refer to these techniques as \emph{Douglas-Rachford type} algorithms.  
These algorithms enjoy good empirical performance but their properties, when applied to the non-convex problem of phase retrieval, are generally not understood. 

In this work, we focus on the RRR algorithm as a representative example of the Douglas-Rachford type algorithms.
In Section~\ref{sec:DR_GD}, we show that in some cases, RRR can be viewed as gradient descent on a certain objective function, all of whose critical points are solutions. 
The intriguing objective function is very different from the  objective functions employed for the random phase retrieval problem.
In particular, the RRR solutions are not minimizers of the objective function.
We also show that in other cases, RRR is not a gradient descent for any objective function.  
Using the underlying objective function of RRR, we give simple proofs of a few basic theoretical results in Section~\ref{sec:analysis}. Specifically, we characterize the set of solutions, show a local convexity around any solution, and derive some stability guarantees. 

\section{The phase retrieval problem and applications} \label{sec:phase_retrieval}

The phase retrieval problem entails finding a signal in the intersection of two sets $x_0\in \A \cap\B$.
We therefore define projectors onto these sets; for the algorithms we consider to be practical, the projectors should be efficiently computed.
For a general $x\in \CC^n$, let $y = Ax\in \CC^m$.  
We consider projectors in terms of $y$ rather than $x$ as the projector onto $\B$ is much cheaper to compute~\cite[Section~4.1]{Li2017a}. 
The projector of $y$ onto the set $\B$ is defined by 
\begin{equation*}
P_\B(y) = b\odot\phase(y),
\end{equation*}
where $b$ is the measured magnitudes~\eqref{eq:pr}, $\odot$ denotes the point-wise product, and the phase operator is defined element-wise as 
\begin{equation*}
\phase(y)[i]:=\frac{y[i]}{|y[i]|}, \qquad y[i]\neq 0,
\end{equation*}
and zero otherwise. The projector onto $\A$, denoted by $P_\A$, is application-specific; a few examples are provided below.

In what follows, a solution is defined as a point whose projections onto the two sets $\A$ and $\B$ are equal (so either projection is in $\A\cap\B$):
\begin{defn}
	\label{defn:sol} A point $y_{0}\in\CC^m$ is said to \emph{correspond to a solution} if $P_\A(y_0) = P_{\B}(y_0)$.
\end{defn}
\noindent We denote a signal that corresponds to a solution by $x_0$ so that $y_0=Ax_0$.
 Importantly, this work focuses on noiseless problems, when exact solutions exist. 
In practice, the data is always contaminated by noise and the definition of a solution should be modified accordingly. In addition, this work considers only discrete setups, and thus neglects  sampling implications. 

We now describe a few specific phase retrieval problem setups and algorithms. 
  In the random phase retrieval problem, the entries of the sensing matrix are usually drawn i.i.d.  from a  normal distribution with $m> 2n$.
A point  $y_{0}\in\CC^m$ that corresponds to a solution should be within the column space of the matrix $A$, that is, $y_0 = AA^\dagger y_0$, where $A^\dagger$ is the pseudo-inverse of $A$; thus,  the set $\A$  describes all signals that lie in the column space of $A$.  
Since this linear projector onto a subspace will be used successively throughout the paper, we denote it by $P_A$, rather than  $P_\A$  which is used for a general (not necessarily linear) projector. In particular, the projection of $y\in\CC^m$ onto the column space of $A$ is given by:
\begin{equation*}
P_A(y) = AA^\dagger y.
\end{equation*} 
It was shown that under rather mild conditions the intersection $\A\cap\B$ is a singleton up to an unavoidable global phase ambiguity: if $y_{0}\in\A\cap\B$ then also $e^{i\theta}y_{0}\in\A\cap\B$ for any global phase~$\theta$; see for instance~\cite{balan2006signal, Bandeira2014,eldar2014phase, Conca2015}.

Since the phase retrieval problem involves searching for an intersection of two sets, and applying each projection separately is cheap, it is natural to apply the two projectors successively; this scheme is called the alternating projections algorithm  and its iterations read: 
\begin{equation}\label{eq:GS}
y\mapsto P_\A P_{\B}(y).   
\end{equation}
In the phase retrieval literature, this technique is usually referred to as Grechberg-Saxton (GS)~\cite{gerchberg1972practical} or error reduction.
The GS algorithm works quite well for the random phase retrieval problem and enjoys supporting theory~\cite{netrapalli2013phase,pauwels2017fienup,waldspurger2018phase,zhang2019phase}, however, in more realistic setups it is known to quickly converge to suboptimal local minima. In practice, it is merely used to refine a solution~\cite{Elser2017, Marchesini2007}. 

A different approach is based on first-order gradient algorithms.
The underling idea is very simple: finding a signal $x\in\CC^n$ that best fits the observed data $b$, that is,
\begin{equation} \label{eq:ls}
\arg\min_{x\in\CC^n} \| b - |Ax|\|^2.
\end{equation}
To minimize~\eqref{eq:ls}, different gradient-based algorithms were applied, equipped with  guarantees on their sample and computational complexities; see for instance~\cite{chen2017solving,cai2016optimal,wang2017solving,Chen2019}.
This approach is flexible and can be combined with different regularizers (e.g., sparsity-promoting terms), and different optimization strategies.  
Gradient-based algorithms were also proposed for other phase retrieval applications in which there are more measurements than unknowns, such as ptychography\footnote{In practice, the sensing matrix in ptychography is not precisely known, making the problem even more challenging.} and frequency-resolved optical gating~\cite{yeh2015experimental,bian2015fourier,bendory2017non,xu2018accelerated,bostan2018accelerated,pinilla2019frequency}. 

We now turn our attention to phase retrieval problems that appear in applications.
In coherent diffraction imaging (CDI), an object is illuminated with a
coherent wave and the  diffraction intensity pattern (equivalent to the Fourier magnitudes of the signal) is measured; thus,  the sensing matrix $A$ is the DFT matrix. As an additional prior, usually  the support of the signal is assumed to be known, (i.e.,\ the signal is known  to be zero outside of some region)~\cite{shechtman2015phase,Bendory2017}. This condition is equivalent to replacing the DFT matrix ($m=n$) with an over-sampled Fourier matrix ($m>n$). Hence, the projector $P_A$ projects $y$ into the column space of the over-sampled DFT matrix. 
In dimension greater than one (as the problem appears in practice), if the over-sampling factor $m/n$ is at least two in each  dimension, then it is known that the solution is unique up to ambiguities~\cite[Corollary 2]{Bendory2017}. However, it was recently shown that this solution might be highly sensitive to perturbations and inexact support knowledge~\cite{barnett2018geometry}.   

In X-ray crystallography, the signal represents the atomic structure of the underlying object, for instance, a 3-D molecular structure. In that case, the signal is sparse, and its $k$ non-zero values correspond to atoms. The measured data is again equivalent to the Fourier  magnitudes of the signal.
Consequently, the sensing matrix $A$ is the DFT matrix, and the set $\A$ describes all signals for which the number of  non-zero values in the signal is at most $k$. The projection onto this set is simply given by keeping the entries corresponding to  the $k$ largest absolute values of the signal, and zeroing out all other entries. In particular, this projection $P_\A$ is not linear.
The solution for the crystallography problem is defined up to three intrinsic ambiguities: multiplication by a complex exponential, shift, and reflection through the origin.

For the last two applications above,  the alternating projection technique and gradient-based methods generally fail to produce meaningful solutions: they tend to quickly convergence to a suboptimal local minimum, far from a point that corresponds to a solution. 
Instead, a family of algorithms that can be described as generalizations of the Douglas-Rachford scheme are employed in practice. The following section introduces this framework and its variants. 

\section{Douglas-Rachford and its generalizations}  \label{sec:DR}
Suppose we wish to solve the minimization problem $\min_yF(y)$, where $F$ is convex. A point $y$ is a minimizer of a function $F(y)$ if and only if $0\in\partial F(y)$, where
\begin{equation}\label{eq:subgrad_defn} 
\partial F(y) = \{v: F(z)\geq F(y) + \langle v, z-y\rangle\ \text{for all } z\},
\end{equation}
is the subdifferential of $F$ at $y$. This is equivalent to requiring $y$ be a fixed point of the \emph{resolvent operator}  $\mathcal{R}_F:=(\lambda\partial F+\II)^{-1}$, for any scalar $\lambda$ \cite[Lemma~2]{Eckstein1992}. Note that even though $\lambda\partial F+\II$ is a set-valued operator (i.e.,\ it is multi-valued), its resolvent is single-valued when $F$ is convex; thus, $\mathcal{R}_F$ is a well-defined function~\cite[Corollary~2.2]{Eckstein1992}.
In fact, for convex $F$ the resolvent $\mathcal{R}_F(y)=\text{prox}_{\lambda F}(y)$ is the proximal mapping of $\lambda F$~\cite[Section~3.2]{boyd_prox_algs}, defined by $$\text{prox}_{\lambda F}(y):=\arg\min_{z\in\RR^n}\left\{\lambda F(z)+\frac{1}{2}\|z-y\|_2^2\right\}.$$  

Now suppose that $F(y)=f(y)+g(y)$ is a sum of two  functions.  This is the case for  phase retrieval since the feasibility problem of finding a point $y\in\mathcal{A}\cap\mathcal{B}$  can be written as $\min_y\mathbb{I}_{\mathcal{A}}(y)+\mathbb{I}_{\mathcal{B}}(y)$, where $\mathbb{I}_\MM$ is the indicator function~\cite{lindstrom2018survey}:
\begin{equation*}
\mathbb{I}_\MM (y)= \begin{cases}
0& \quad y\in \MM, \\  \infty& \quad y\notin \MM. 
\end{cases}
\end{equation*}
The indicator function $\mathbb{I}_\MM$ is convex if and only if $\MM$ is convex.
In many cases, computing $\mathcal{R}_f$ and $\mathcal{R}_g$ individually might be cheap, while computing $\mathcal{R}_{f+g}$ is expensive. For example, the resolvent of an indicator function $\mathcal{R}_{\mathbb{I}_\MM}$ is just the projection operator onto $\MM$ (for any~$\lambda$). Therefore, applying $\mathcal{R}_{\mathbb{I}_{\A}}$ and $\mathcal{R}_{\mathbb{I}_{\B}}$ amounts to projecting onto the sets $\A$ and $\B$, which can be done cheaply as in Section~\ref{sec:phase_retrieval}. On the other hand, applying $\mathcal{R}_{\mathbb{I}_{\A}+\mathbb{I}_{\B}}$ amounts to projecting onto $\A\cap\B$, which is equivalent to solving the phase retrieval problem. 
If $f$ and $g$ are convex functions, there is a simple way to formulate the problem of finding $y=\mathcal{R}_{f+g}(y)$ only in terms of the operators $\mathcal{R}_{f}$ and $\mathcal{R}_{g}$.
To this end, let us define the \emph{Cayley operator} associated with $\MM$ by  $\mathcal{C}_\MM:=2\mathcal{R}_\MM-\II$. Then, we have: 
\begin{prop}\label{prop:correspondence_convex}
Suppose $f$ and $g$ are convex functions. Then, $y = \mathcal{R}_{f+g}(y)$ if and only if $\mathcal{C}_f\mathcal{C}_g(z)=z$, where $y=\mathcal{R}_{g}(z)$. 
\end{prop}
\begin{proof} 
Let us assume that $\mathcal{C}_f\mathcal{C}_g(z)=z$, which is equivalent to $\mathcal{R}_f\big(2\mathcal{R}_g(z)-z\big) = \mathcal{R}_g(z)$. Since we defined $y=\mathcal{R}_g(z)$, this can be rewritten as $\mathcal{R}_f(2y-z) = y$.
Now, $\mathcal{R}_f(2y-z)=y$ if and only if $2y-z\in y+\lambda\partial f(y)$, and $y=\mathcal{R}_g(z)$ if and only if $z\in y+\lambda\partial g(y)$. Adding these two properties together yields
\begin{equation} 
\mathcal{R}_f(2y-z) = y \implies 2y\in 2y+\lambda(\partial f(y)+\partial g(y)). 
\end{equation}
This is equivalent to $0\in \partial f(y)+\partial g(y)$, and thus $y=\mathcal{R}_{f+g}(y).$

Conversely, let us assume that $y=\mathcal{R}_{f+g}(y)$ and therefore $0\in \partial f(y)+\partial g(y)$. Then, since also $z\in y+\lambda\partial g(y)$ by definition of $y$, we can subtract the two to obtain $2y-z\in y+\lambda\partial f(y)$, hence $\mathcal{R}_f(2y-z)=y$ and $\mathcal{C}_f\mathcal{C}_g(z)=z$ is a fixed point.
\end{proof}

\noindent Proposition~\ref{prop:correspondence_convex} implies that in the convex case it suffices to find a fixed point for $\mathcal{C}_f\mathcal{C}_g$, which involves computing only $\mathcal{R}_f$ and $\mathcal{R}_g$. Naively, we may attempt to apply the fixed-point iterations $y\mapsto \mathcal{C}_f\mathcal{C}_g(y)$. Unfortunately, this is not guaranteed to converge even if both $f$ and $g$ are convex \cite[Section~4.1]{GS_counterexample}. Instead, the Douglas-Rachford algorithm iterates
\begin{equation}\label{eq:Doug_Rach}
y\mapsto \frac{1}{2}(\II+\mathcal{C}_{f}\mathcal{C}_{g})(y),
\end{equation} 
which is guaranteed to converge in the convex case whenever a solution $0\in\partial F(y)$ exists \cite[Section~3]{GS_counterexample}.

Generally, while the success of the Douglas-Rachford algorithm for closed, convex sets is well-understood, very little is known for the non-convex setting; see~\cite{lindstrom2018survey,li2016douglas} and references therein. 
	For instance, a local linear convergences  for non-convex sets was proven under several conditions~\cite{hesse2013nonconvex,phan2016linear}, however, it is not clear whether these conditions hold for phase retrieval.  
	In addition, the sequence generated by Douglas-Rachford is generally known to be bounded~\cite[Theorem 4]{li2016douglas}.
	Despite the lack of supporting theory, in practice the Douglas-Rachford type algorithms are known to solve   challenging non-convex problems, such as the Diophantine equations, bit retrieval, sudoku, and protein conformation determination~\cite{elser2007searching,borwein2017reflection,Elser2018}. 
	In addition, even for the random phase retrieval for which gradient-based algorithms were studied thoroughly, it was demonstrated numerically that Douglas-Rachford outperforms these gradient-based alternatives when the number of measurements drops close to the information-theoretic limit~\cite[Appendix~A]{Elser2017}. 

\subsection{Douglas-Rachford for phase retrieval}

As stated in the preceding section, the phase retrieval feasibility problem of finding $y\in\A\cap\B$ can be written as $\min_y\mathbb{I}_\A(y)+\mathbb{I}_\B(y)$ and the resolvents of these two indicator functions are simply the projections onto the two constraint sets $\mathcal{R}_{\mathbb{I}_\A}=P_\A$ and $\mathcal{R}_{\mathbb{I}_\B}=P_\B$. The corresponding Cayley operators are the \emph{reflections} across these sets: $\mathcal{C}_{\mathbb{I}_\A}:=R_\A=2P_\A - \II$ and similarly for $\B$.
Therefore, the Douglas-Rachford iterations for for phase retrieval read
\begin{equation}\label{eq:Doug_Rach1}
    y\mapsto \frac{1}{2}(I+R_\A R_{\B})(y) = y + P_\A\big(2P_\B(y)-y\big) - P_\B(y).
\end{equation} 
The $(t+1)^{th}$ iteration of the algorithm can  be parse as
\begin{align*}
y_1 &= P_\B (y^{(t)}) \\
y_2 &= P_\A (2y_1 - y^{(t)})\\
y^{(t+1)} &= y^{(t)} + (y_2-y_1) .
\end{align*}
This formulation unveils close relations with the method of alternating direction
method of multipliers (ADMM)~\cite{boyd2011distributed}; see for instance~\cite{Elser2017a}.

Unfortunately, the set $\B$~\eqref{eq:setB} is not convex  and thus Proposition~\ref{prop:correspondence_convex} and the derivation preceding it does not apply for phase retrieval. If it were true, Proposition~\ref{prop:correspondence_convex} would imply that $y$ is a fixed point of the iterations~\eqref{eq:Doug_Rach1} if and only if $P_\B(y)\in\A\cap\B$ is a solution. This is false. For a trivial counterexample, consider $y_0=Ax_0$ such that $y_0\neq 0$ for some $A$ and $x_0$. Let $y=y_0/2$. In this case, $P_\B(y)=y_0\in \A\cap \B$ is a solution, but $P_\A(y)=y_0/2\neq P_\B(y)$, thus $y$ does not correspond to a solution. 
Nevertheless, the Douglas-Rachford iterations are still well-defined. If $\A$ is a linear subspace, the following proposition shows that Douglas-Rachford stops only when a solution is found:
\begin{prop} \label{prop:correspondence}
	If $\A$ is a linear subspace, then $y$ is a fixed point of the iterations~\eqref{eq:Doug_Rach1} if and only if $y$ corresponds to solution. 
\end{prop}
\begin{proof}
If the iterations stagnate, then $P_\A(2P_{\B}- I)(y) = P_{\B}(y)$. 
Applying $P_\A$ on both sides yields $P_\A(2P_{\B}- I)(y) = P_\A P_{\B}(y)$ and using the linearity of $P_\A$, we get $P_\A P_{\B}(y) = P_\A(y)$. Applying $(I-P_\A)$ yields $0 = (I-P_\A)P_{\B}(y)$ and thus $P_\A P_{\B}(y)=P_{\B}(y)$. Therefore, $P_\A(y)=P_\B(y)$.
Conversely, if~$y$ corresponds to a solution then by definition $P_\A(y)=P_\B(y)$ and thus $P_\A(y)=P_\A P_\B(y)$. Therefore,
\begin{equation*}
2P_\A P_\B(y) = 2 P_\A(y) = P_\A(y) + P_\B(y).
\end{equation*}
By the linearity of $P_\A$, we have then have $P_\A(2P_{\B}- I)(y) = 2P_\A P_\B(y)-P_\A(y)=P_\B(y)$ so $y$ is a fixed point.
\end{proof}
Note that if $y$ is a fixed point then $P_{\B}(y)\in\A\cap\B$ is a solution, but as the example before the proposition shows, the converse is false. Also, Proposition~\ref{prop:correspondence} does not guarantee that Douglas-Rachford actually converges to a solution, even when $\A$ is a linear subspace, and as far as we know this is still an open problem.

If $\A=\text{col}(A)$ is the subspace spanned by the columns of $A$, then it is straight-forward to show that the iterations~\eqref{eq:Doug_Rach1} can be rewritten as
\begin{equation}\label{eq:Doug_Rach2}
y\mapsto  P_A P_{\B}(y) + P_A^cP_{\B}^c(y),
\end{equation} 
where $P_\MM^c:=I-P_\MM$. 
This formulation offers an interesting interpretation of the Douglas-Rachford iterations. 
The first term is precisely the GS iterations~\eqref{eq:GS}, which tend to get trapped in irrelevant stagnation points.
The second term moves in the orthogonal complement of the column space, and its addition  guarantees that all the fixed points of the Douglas-Rachford scheme correspond to solutions.  

\subsection{Generalizations of Douglas-Rachford}

 Many algorithms proceed to relax these iterations by introducing different free parameters: 
  \begin{itemize}
    \item Fienup's hybrid input-output (HIO) algorithm proceeds by iterating 
       \begin{equation*}
    y\mapsto y + P_\A\big((1+\beta)P_\B(y)-y\big)-\beta P_\B(y).
    \end{equation*} 
    If $\A$ is linear, then it can be also written as
    \begin{equation}\label{eq:HIO}
        x\mapsto P_A P_{\B}(y) + P_A^c(I-\beta P_{\B})(y),
    \end{equation} 
    where $\beta$ is a parameter controlling the ``negative feedback.''
     \item The relaxed reflect reflect (RRR) algorithm iterates
    \begin{equation*}
        y\mapsto y + \beta\left(P_\A\big(2P_{\B}(y)-y\big)-P_\B(y)\right),
    \end{equation*}
    which, if $\A$ is linear, can be rewritten as
    \begin{equation}\label{eq:RRR} 
    y\mapsto y + \beta(2P_AP_\B(y)-P_A(y)-P_\B(y)) = (1-\beta)y + \beta\Big(P_AP_\B(y)+P_A^cP_\B^c(y)\Big).
    \end{equation}
    If $\beta\in(0,1)$, this is a convex combination of $y$ and the Douglas-Rachford iterate~\eqref{eq:Doug_Rach2}.
    \item The relaxed averaged alternating reflections (RAAR) algorithm iterates
    \begin{equation*}
        y\mapsto \beta\left(y + P_\A\big(2P_{\B}(y)-y\big)\right) + (1-2\beta)P_{\B}(y),
    \end{equation*}
    which if $\A$ is linear can be rewritten as
    \begin{equation}\label{eq:RAAR}
        y\mapsto \beta\left(y + 2P_\A P_{\B}(y)-P_\A(y)-P_{\B}(y)\right) + (1-\beta)P_{\B}(y),
    \end{equation}
    and interpreted as another convex combination if $\beta\in(0,1)$.
\end{itemize}
Clearly, if $\beta=1$ then all these algorithms coincide with the Douglas-Rachford scheme. Many other variants exist in the literature; see for instance~\cite{lindstrom2018survey} and references therein. 
For any $\beta>0$, when $\A$ is linear the RRR and HIO iterations 
stall only when a solution is obtained:
\begin{cor}
	\label{lem:any_fixed[i]s_good} When $\A$ is a linear subspace, $y$ is a fixed point of RRR or HIO if and only if $y$ corresponds to solution.
\end{cor}
\begin{proof}
	The proof follows from the proof of Proposition~\ref{prop:correspondence}.
\end{proof}


\section{RRR as a gradient algorithm}
\label{sec:DR_GD}

Let us consider the following objective function
\begin{equation} \label{eq:fr}
\fR(y) = \|y-P_\A P_\B(y)\|_2^2 - \frac{1}{2}\left(\|y-P_\A(y)\|_2^2 + \|y-P_\B(y)\|_2^2\right).
\end{equation}
Assuming $y$ and $A$ are real, and $P_A$ is a linear  projection onto the column space of $A$, 
the gradient of $\fR$ at any point $y$ none of whose coordinates are zero reads:
\begin{equation} 
\nabla \fR(y) = P_\A(y) + P_\B(y) - 2P_\A P_\B(y). 
\end{equation}
That is because $\nabla P_{\mathcal B}(y) = 0$ at any point $y$ such that $y_i\neq 0$ for all $i$, which follows because the sign function is locally constant.
Therefore, as long as none of the iterates of RRR~\eqref{eq:RRR} have a zero coordinate, RRR can be viewed as gradient descent on $f_R$, whose iterations are: 
\begin{equation} \label{eq:grad_flow}
y_{t+1} = y_t - \beta \nabla \fR(y_t).  
\end{equation}

Note that GS~\eqref{eq:GS} is also a gradient algorithm, with a constant step size, when the underlying objective function is 
\begin{equation}
f_{\text{GS}}(y) = \frac{1}{2}\|y-P_\A P_\B(y)\|_2^2.
\end{equation}
Therefore, the RRR iterations balance between two opposite forces:  RRR tries to minimize  $f_{\text{GS}}$, while at the same time it aims to maximize the $\ell_2$ distance of $y$ from the sets $\A$ and $\B$. 
A similar observation was made by Marchesini, who formulated HIO~\eqref{eq:HIO} as an instance of  saddle-point optimization~\cite{Marchesini2007}.
Nevertheless, searching for a saddle-point might by an unstable process, whereas our formulation allows us to derive some stability guarantees. For example, in Proposition~\ref{lem:local_convex} we show that $\fR$ is strongly convex in a small region around a solution.


While~\eqref{eq:grad_flow}  establishes an interesting connection between RRR and  the gradient-based algorithms solving~\eqref{eq:ls}, there is a notable difference between the two approaches. 
In gradient-based algorithms, we usually aim to minimize an objective function by setting its gradient to zero.
For RRR, we also wish to find a zero of the gradient and the objective, however, crucially, the solution is not a minimizer of the objective function $\fR$:  $\fR(y)<0$ for any suboptimal fixed point of GS, since such a point satisfies $y=P_\A P_\B(y)$ but  $y\neq P_\A(y)$ or $y\neq P_\B(y)$; this is an unusual scenario from optimization point-of-view.
 Moreover, attempting to run gradient descent on $\fR$ with standard optimization techniques, e.g., backtracking linesearch, will result in the chosen step sizes to rapidly go to zero, so the algorithm is effectively stuck at suboptimal points that are not even critical points. Therefore, in practice the RRR algorithm is run with a constant step size. 
Figure~\ref{fig:example} shows an example of signal recovery from the random phase retrieval setup. The objective function oscillates and drops below zero many times, until at some point it convergences quickly to the solution.

\begin{figure}[ht]
	\begin{subfigure}[h]{0.45\textwidth}
		\centering
		\includegraphics[scale=.6]{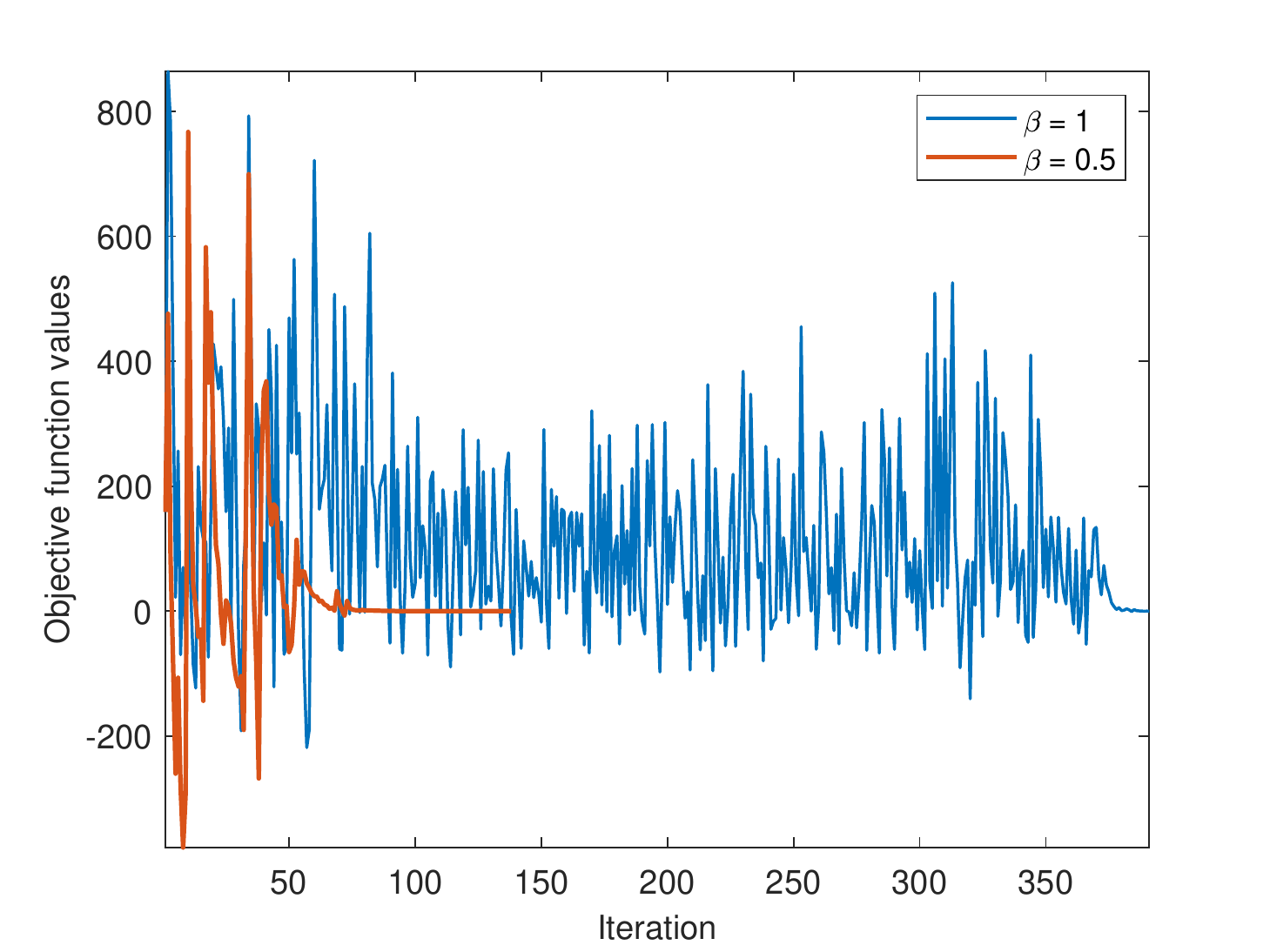}
		\caption{\label{fig:sr}}
	\end{subfigure}
	\hfill
	\begin{subfigure}[h]{0.45\textwidth}
		\centering
		\includegraphics[scale=.6]{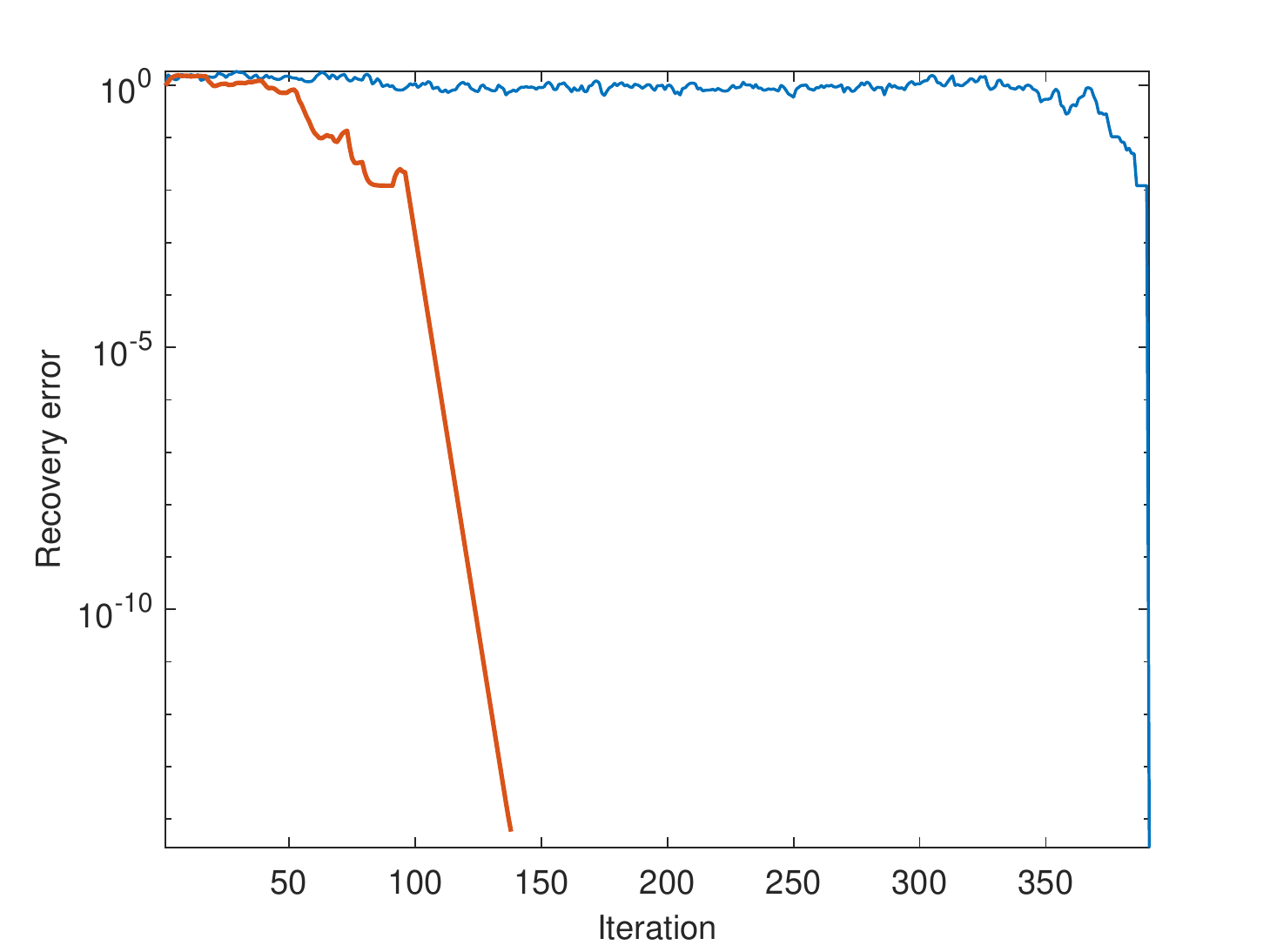}
		\caption{}
	\end{subfigure}
	\caption{\label{fig:example} Signal recovery from the random phase retrieval problem with RRR. A signal of length $n=50$  and a sensing matrix $A\in\RR^{80\times 50}$ were drawn from an i.i.d.\ normal distribution.  The signal was recovered with the RRR algorithm~\eqref{eq:RRR} with $\beta=0.5$ (in red) and $\beta=1$ (in blue); the latter coincides with the Douglas-Rachford scheme~\eqref{eq:Doug_Rach}. \textbf{Figure (a)} shows the non-monotonic progress of  the objective function~\eqref{eq:fr} with the iterations;  the values drop below zero multiple times. \textbf{Figure (b)} presents the recovery error with respect to the true signal when taking the sign ambiguity into account. This measure is not available in practice.}  
\end{figure}

Unfortunately, the analysis presented in this section is restricted to the case where $A$ and $x_0$ are real. This is a major drawback since in practice the matrix $A$ in phase retrieval applications is  complex. The following result shows that in the complex case, RRR is a not gradient descent for any objective function.
\begin{prop} Suppose that $P_A$ is a linear projection onto the column space of $A$. Then, if $y$ is a complex variable and $y[i]\neq 0$ for all $i$, then RRR is not gradient descent for any objective function.
\end{prop}
\begin{proof}
As shown in \cite{Marchesini2007}, the operators  $P_A(y),P_{\B}(y)$ are indeed gradients. However, $P_AP_{\B}(y)$ is not a gradient.
To see that, we compare the mixed Wirtinger derivatives:
	\[\begin{aligned} \frac{\partial}{\partial \overline{y[k]}}P_AP_{\B}(y)[i] &= -\frac{1}{2}(AA^{\dagger})[i,k]|b[k]|\frac{y[k]}{|y[k]|\overline{y[k]}},\\
	\frac{\partial}{\partial \overline{y[i]}}P_AP_{\B}(y)[k] &= -\frac{1}{2}(AA^{\dagger})[k,i]|b[i]|\frac{y[i]}{|y[i]|\overline{y[i]}}, \end{aligned}\]
	so $\frac{\partial}{\partial \overline{y[k]}}P_AP_{\B}(y)[i]\neq \frac{\partial}{\partial \overline{y[i]}}P_AP_{\B}(y)[k]$. 
	If $P_AP_{\B}$ was a gradient, then the Hessian (of the underlying function) would have been a symmetric matrix; this is not the case here. 
		In the real case, the derivative of the sign function is zero (besides at the origin) and thus the equality. 
\end{proof}

\section{Analysis}
\label{sec:analysis}

We derive several basic results about RRR  by viewing it as gradient descent~\eqref{eq:grad_flow}. In what follows, we consider the case where both $A$ and $x_0$ are real, so $\phase(y)=\text{sign}(y)$ for all signals $y$, and let $P_\A$ be the projection onto the column space of $A$:
\begin{equation*}
P_\A(y) = P_A(y) = AA^\dagger y.
\end{equation*}
\noindent We denote the column space of $y$ by $\col(A)$ and its orthogonal complement by $\col(A)^{\perp}$; the projection onto $\col(A)^{\perp}$ is given by ($\II-P_A$).
 The $i$th entry of a signal $z$ is denoted by $z[i]$.
 All proofs are provided in Section~\ref{sec:proof}. 


\paragraph{Set of solutions.} 
 The first result characterizes the set of fixed points of RRR.

\begin{prop}
\label{lem:solns_char} 
A signal $y$ corresponds to a solution if and only if $y = \tilde y + w$, where $ \tilde y\in\col(A)\cap \B$ and $w\in\col(A)^{\perp}$, satisfying either:
\begin{enumerate}
    \item $w[i]=0$; or,
	\item $\text{sign}(w[i]) = \text{sign}(\tilde y[i])$; or, 
	\item $\text{sign}(w[i]) = -\text{sign}(\tilde y[i])$ and $|w[i]|<|(Ax_0)[i]|$,
\end{enumerate}
for each $1\leq i\leq m$.
\end{prop}

\paragraph{Local convexity.} 
In~\cite{Li2017a}, it was proven that   if $A\in\CC^{m\times n}$ with $m/n\geq 2$ is isometric, $\beta\in(0,1]$, and  $y$ is sufficiently close to $y_{0}\in\col(A)\cap\B$, then RAAR~\eqref{eq:RAAR} converges linearly to $y_{0}$. 
In the real case, we prove a stronger result. In particular, the following proposition shows that every fixed point is a local minimum of $\fR(y)$ around which $\fR(y)$ is  convex, making our formulation more stable than the saddle-point formulation in~\cite{Marchesini2007}.   
\begin{prop}
\label{lem:local_convex} 
Suppose that $y_{0}$ corresponds to a solution and that $d = \min_i|y_0[i]| > 0$. Then:
\begin{enumerate}
	\item $\fR$ is convex in the $\ell_2$ ball of radius $d$ about $y_{0}$, and 1-strongly convex\footnote{Recall that a function $f$ is $p$-strongly convex if, for all $u,v$ in its domain, the following inequality holds $\left(\nabla f(u) -\nabla f(v)\right)^T(u-v)\geq p\|u-v\|^2_2$.} when restricted to the intersection of this ball with $\col(A)$; 
	\item if $||y-y_{0}||_2<d$ and $\beta\in(0,2)$, then RRR converges to a fixed point linearly; if $\beta=1$ (that is, RRR coincides with Douglas-Rachford), then  only  one iteration is required.
\end{enumerate}
\end{prop}
A similar result was shown for bit retrieval in~\cite[Section~VII]{Elser2018}.
\paragraph{Stability.} 
Next, we show a stability result: if the norm of the gradient of $\fR$ is sufficiently small then there is a solution nearby. 
\begin{prop}
	\label{lem:stable} There exists $\epsilon>0$, such that if $||\nabla \fR(y)||_2< \epsilon$ then:
	\begin{enumerate}
		\item $P_{\B}(y)\in\col(A)\cap\B$ is a solution;
		\item if $d=\min_i|(Ax_0)[i]|>0$, then there exists a point $y_0\in\RR^m$ that corresponds to a solution such that $||y-y_0||_2<\epsilon\left(1 + \frac{||P_A^c(y)||_2}{d}\right)$; 
		\item if in addition $\min_i|y[i]|\geq\epsilon$, then $||y-y_0||_2< \epsilon$.
	\end{enumerate}
\end{prop}
\noindent Note that (1) does not imply that $y$ corresponds to a solution (or equivalently, that $\nabla \fR(y)=0$), as shown in Section~\ref{sec:DR}. In addition, (2)  does not claim that all $y$ near $y_0$ will converge to $y_0$; this is true under the stronger assumptions of Proposition~\ref{lem:local_convex}.

We are trying to find a zero of both $\fR$ and $\nabla\fR$ by gradient descent, while $\fR$ itself can become negative. 
The next lemma shows that $\fR(y)$ becomes positive for large enough step size along almost any search direction from any point: 
\begin{lemma}
	\label{lem:no_escape} For any $y\in\RR^m$ and any direction $d\in\RR^m$, such that $d[i]\neq 0$ for all $i$ and either $P_A(d)\neq 0$ or $\langle d, P_{\B}(y)\rangle > 0$, we have $\lim_{\beta\to\infty}\fR(y-\beta d) = \infty$. 
\end{lemma}

The following corollary states that if the gradient is non-vanishing, then it satisfies the conditions on the direction $d$ of Lemma~\ref{lem:no_escape}. In other words, the negative of the gradient is a good direction to follow. 
\begin{cor}
	\label{cor:RRR_no_escape} For any $y\in\RR^m$ such that $\nabla \fR(y)[i]\neq 0$ for any $i$, there exists a sufficiently large step size $\beta>0$ such that $\fR(z-\beta\nabla \fR(y)) > 0$. 
\end{cor}

%
%

\section{Discussion}

This work is part of ongoing efforts to explain the remarkable effectiveness of Douglas-Rachford type algorithms for phase retrieval, as well as other non-convex hard problems. 
In particular, we have shown that RRR can be viewed, in some cases, as gradient descent on a certain objective function. The solutions are critical points of that objective.
This relates Douglas-Rachford with the vast body of literature about gradient techniques for the random phase retrieval problem.
However, in contrast to the common practice in optimization, the objective function can take negative values and therefore a solution is not a minimizer of the objective.  

Using the objective function~\eqref{eq:fr}, we have derived new results that establish local convexity in the vicinity of a solution and show that the solutions are stable (in the sense of Proposition~\ref{lem:stable}). We hope to harness recent exciting results on first-order methods in different non-convex settings (see for instance~\cite{lee2016gradient,sun2016complete,boumal2016nonconvex,balakrishnan2017statistical,boumal2018deterministic,li2019rapid} , just to name a few) to extend these results and unveil the  global properties of RRR.
One particular goal is to understand the source and basic characteristics of the dynamical behavior of RRR far from any solution, as demonstrated in Figure~\ref{fig:example} and in~\cite{Elser2017}. 

\section{Proofs} \label{sec:proof}
%

\subsection{Proof of Proposition~\ref{lem:solns_char}}\label{sec:pf_solns_char}
Let $y = \tilde y+w$ with $\tilde y,w$ as hypothesized.  Then, $P_A(y)=\tilde y$ and $P_{\B}(y) = P_{\B}(\tilde y)=\tilde y$. The last equality holds because $\tilde y\in \text{col}(A)\cap B$ by hypothesis, while the first equality holds since either
$$\phase(\tilde y[i]+w[i]) = \phase((|(Ax_0)[i]|+|w[i]|)\phase(\tilde y[i]))=\phase(\tilde y[i]),$$ or $$\phase(\tilde y[i]+w[i]) = \phase((|(Ax_0)[i]|-|w[i]|)\phase(\tilde y[i]))=\phase(\tilde y[i]).$$ Therefore, $P_A(y)=P_{\B}(y)$.

Conversely, if $y$ corresponds to a solution, then we can write $y = P_A(y) + P_{A^c}(y)$. Since $P_{\B}(y) = P_A(y) = P_{\B}P_A(y)$,  we have  $$\phase(P_A(y)[i] + P_{A^c}(y)[i]) = \phase(P_A(y)[i]).$$ Therefore, either $P_{A^c}(y)[i]=0$ or $\phase(P_{A^c}(y)[i])=\phase(P_A(y)[i])$ or $\phase(P_{A^c}(y)[i])=-\phase(P_A(y)[i])$ and $|P_{A^c}(y)[i]|<|P_{A}(y)[i]|$, as desired.

\subsection{Proof of Proposition~\ref{lem:local_convex}}\label{sec:pf_local_convex}
Note that if $\text{sign}(y[i])\neq \text{sign}({y_{0}}[i])$ for some $i_0$, then
\[ ||y-y_{0}||_2^2 = \sum_i|y[i]-{y_{0}}[i]|^2 \geq (|y[i_0]| + |{y_{0}}[i_0]|)^2 \geq |{y_{0}}[i_0]|^2.\]
Therefore,  if $||y-y_{0}||_2<d$ we must have $\text{sign}(y[i])=\text{sign}({y_{0}}[i])$ for all $i$ and hence $$P_{\B}(y)=P_{\B}(y_{0})=P_A(y_{0}).$$ 
Therefore, in this $\ell_2$ ball the objective function simplifies to
\[\fR(y) = \frac{1}{2}\left(||y-P_A(y_{0})||_2^2 - ||(I-P_A)(y)||_2^2\right),\]
so $f(y)$ is infinitely differentiable. 
Then, we have 
\begin{equation} \label{eq:lemma_convex_grad}
\nabla \fR(y) = P_A(y-y_{0}),
\end{equation}
and $\nabla^2 \fR(y) = AA^{\dagger}\succeq 0$, so $\fR(y)$ is convex. Furthermore, when restricted to $\col(A)$ all the eigenvalues of $AA^{\dagger}$ are 1 as it is a projection matrix onto $\col(A)$, so $\fR(y)|_{\col(A)}$ is 1-strongly convex.
This concludes the proof of the first part.

Next, let us assume that  $||y-y_{0}||_2<d$ and $\beta\in(0,2)$. 
According to~\eqref{eq:lemma_convex_grad}, the $t^{\text{th}}$ RRR iteration reads 
\begin{equation*}
\begin{split}
y^{(t)} &= y^{(t-1)}-\beta P_A(y^{(t-1)}-y_{0}) \\&= (1-\beta)P_A(y^{(t-1)}) + \beta P_A(y_{0}) + P_A^c(y^{(t-1)}).
\end{split}
\end{equation*}
Therefore, 
\begin{align*}
y^{(t)}-y_0 &= (1-\beta)P_A(y^{(t-1)}) + \beta P_A(y_{0}) + P_A^c(y^{(t-1)})-P_A(y_0)-P_A^c(y_0) \\
&=(1-\beta)P_A(y^{(t-1)}-y_0) + P_A^c(y^{(t-1)}-y_0),
\end{align*}
and
\[\begin{aligned} ||y^{(t)}-y_{0}||_2^2 &= (1-\beta)^2||P_A(y^{(t-1)}-y_{0})||_2^2 + ||P_A^c(y^{(t-1)}-y_{0})||_2^2\\
&< ||P_A(y^{(t-1)}-y_{0})||_2^2 + ||P_A^c(y^{(t-1)}-y_{0})||_2^2\\
&= ||y^{(t-1)}-y_{0}||_2^2\\
&<d.\end{aligned}\]
This implies that if we initialize $y^{(0)}$ such that $||y^{(0)}-y_{0}||_2<d$, and use a constant step size $\beta\in(0,2)$, then $||y^{(t)}-y_0||_2<d$ for all $t\geq 1$ so the RRR iterations stay within this ball, and   
\begin{equation*} 
y^{(t)} = (1-\beta)^tP_A(y^{(0)}) + \left(1-(1-\beta)^t\right)P_A(y_0) + P_A^c(y^{(0)}),
\end{equation*} 
so $$y_{\infty}=\lim_{t\to\infty}y^{(t)} = P_A(y_0) + P_A^c(y^{(0)}) = P_{\B}(y_0) + P_A^c(y^{(0)}).$$ Note that $y_{\infty}$ corresponds to a solution by Corollary~\ref{lem:any_fixed[i]s_good} and the fact that $\nabla \fR(y_{\infty}) = 0$. Also note that if $\beta = 1$, then $y^{(1)}=y_{\infty}$ so RRR converges to $y_{\infty}$ in one iteration.

\subsection{Proof of Proposition~\ref{lem:stable}}
\label{sec:pf_stable}
Note that since
\begin{equation*}
\begin{split}
\nabla \fR(y) &= P_A(y) + P_\B(y) - 2P_A P_\B(y) \\
&= P_A(I-P_\B)(y) + (I-P_A)P_B(y),
\end{split}
\end{equation*}
then
\[||\nabla \fR(y)||_2^2 =  ||P_A(y) - P_AP_{\B}(y)||_2^2 + ||P_{\B}(y) - P_AP_{\B}(y)||_2^2.\]
Therefore, $||P_{\B}(y) - P_AP_{\B}(y)||_2\leq ||\nabla f(y)||_2$ and $||P_A(y) - P_AP_{\B}(y)||_2\leq ||\nabla \fR(y)||_2$. Then note that $||P_{\B}(y) - P_AP_{\B}(y)||_2$ depends only on the signs of $y$, and hence takes at most $2^m$ values, one of which is zero. Therefore, there exists $\epsilon_1$ such that if $||P_{\B}(y) - P_AP_{\B}(y)||_2<\epsilon_1$ then in fact $P_{\B}(y) = P_AP_{\B}(y)$ and so $P_{\B}(y)\in\col(A)\cap\B$ is a solution, which proves the first claim. In this regime $$\nabla \fR(y)=P_A(I-P_\B)(y).$$
Taking $\epsilon\leq \epsilon_1$, we then have
\begin{equation} \label{eq:pa_pab}
||P_A(y) - P_AP_{\B}(y)||_2^2 <\epsilon.
\end{equation}

Let $y_0 = P_{\B}(y) + (1-\alpha)P_A^c(y)$ 
for $\alpha\in(0,1)$. We claim that there exists $\alpha$ depending on $\epsilon$ such that $y_0$ corresponds to a solution. First, we claim that if $\min_i|(Ax_0)[i]|>0$ then $P_{\B}(y) = P_{\B}P_A(y) = P_AP_{\B}(y)$, in which case
\begin{equation} \label{eq:inB}
P_A(y_0)=P_A P_{\B}(y)=P_{\B}(y)\in\B.
\end{equation}

For general vectors $u,v\in\RR^m$, note that if $\text{sign}(u[i_0]+v[i_0])\neq \text{sign}(u[i_0])$ for some $1\leq {i_0}\leq m$, then
\[\begin{aligned} ||u - P_{\B}(u+v)||_2^2 &= \sum_{i=1}^m\Big||u[i]|\text{sign}(u[i]) - |(Ax_0)[i]|\text{sign}(u[i]+v[i])\Big|^2 \\& \geq \Big| |u[i_0]| + |(Ax_0)[i_0]|\Big|^2 \\&\geq |(Ax_0)[i_0]|^2.\end{aligned}\]
Hence, if $\min_i|(Ax_0)[i]|>0$ and we choose $$\epsilon<\min(\epsilon_1, |(Ax_0)[1]|,\ldots,|(Ax_0)[m]|)$$
and substituting $u=P_A(y)$ and $v=P_{A^C}(y)$ we have
\begin{equation*}
||P_A(y)-P_{\B}(P_A(y)+P_{A^C}(y))||_2 = ||P_A(y)-P_{\B}(y)||_2 >\epsilon.
\end{equation*}
Since   $P_{\B}(y)=P_{A}P_{\B}(y)$, 
it contradicts~\eqref{eq:pa_pab} and therefore we must have  
$$\text{sign}(y[i])=\text{sign}(P_A(y)[i] + P_A^c(y)[i]) = \text{sign}(P_A(y)[i]),$$ for all $i$. Consequently, $P_{\B}(y) = P_{\B}P_A(y) = P_AP_{\B}(y)$ which is the desired claim.

Next, we wish to show that $P_{\B}(y_0)=P_{\B}(y)$, or equivalently, $$\text{sign}(P_{\B}(y) + (1-\alpha)P_A^c(y))=\text{sign}(y-w-\alpha P_A^c(y))=\text{sign}(y),$$ where $w = P_A(y) - P_{\B}(y)$; note that $||w||_2 < \epsilon$ and so also $|w[i]|< \epsilon$ for all $i$. 
 Let us define the set
\[\begin{aligned}I &= \left\{i\in\{1,\ldots, m\}:\ \text{sign}(P_A^c(y)[i])\neq \text{sign}(y[i]), \text{sign}(w[i]) = \text{sign}(y[i]),\ |P_A^c(y)[i]|\geq |(Ax_0)[i]|\right\}.\end{aligned}\]
\begin{lemma} \label{lem:not_in_i}
If $i\notin I$, then $\text{sign}(y_0[i]) = \text{sign}(y[i])$.
\end{lemma}
\begin{proof}
Note that if $i\notin I$ then either:
\begin{itemize}
	\item $\text{sign}(P_A^c(y)[i])=\text{sign}(y[i])$: in which case
	\[\begin{aligned} \text{sign}(y_0[i])&=\text{sign}(P_{\B}(y)[i]+(1-\alpha)P_A^c(y)[i])\\& =\text{sign}\left[(|(Ax_0)[i]|+(1-\alpha)|P_A^c(y)[i]|)\text{sign}(y[i])\right]\\&=\text{sign}(y[i]); \end{aligned}\]
	
	\item $\text{sign}(P_A^c(y)[i])=-\text{sign}(y[i])$ or $P_A^c(y)[i]=0$, and $\text{sign}(w[i]) = -\text{sign}(y[i])$: in which case
	\[ \begin{aligned}
	\text{sign}(y_0[i])&=\text{sign}(y[i]-w[i]-\alpha P_A^c(y)[i])\\&=\text{sign}\left[(|y[i]|+|w[i]|-\alpha|P_A^c(y)|)\text{sign}(y[i])\right]\\&=\text{sign}(y[i]),
	\end{aligned}\]
	as $|y[i]|\geq |P_A^c(y)[i]| > \alpha|P_A^c(y)[i]|$ if $P_A^c(y)[i]\neq 0$;
	
	\item $|P_A^c(y)[i]|<|(Ax_0)[i]|$: in which case
	\[\begin{aligned}
	\text{sign}(P_{\B}(y)[i]+(1-\alpha)P_{A}^c(y)[i]) &= \text{sign}[(|(Ax_0)[i]|\pm (1-\alpha)|P_{A^c}(y)|)\text{sign}(y[i])] \\&= \text{sign}(y[i]).
	\end{aligned}\]
\end{itemize}
\end{proof}

Next, note that  $\text{sign}(y[i]) = \text{sign}(P_{\B}(y)[i]+w[i]+P_A^c(y)[i])$.
For $i\in I$, we get $$\text{sign}(y[i]) = \text{sign}\left[(|(Ax_0)[i]| + |w[i]| - |P_A^c(y)[i]|)\text{sign}(y[i])\right];$$ so $$|P_A^c(y)[i]|<|(Ax_0)[i]| + |w[i]|< |(Ax_0)[i]| + \epsilon,$$ and hence
\begin{equation} \label{eq:sandwich}
|(Ax_0)[i]|\leq |P_A^c(y)[i]|<|(Ax_0)[i]| + \epsilon,\quad \forall i\in I.
\end{equation}
Let $d = \min_i|(Ax_0)[i] > 0$ and define $\alpha = \epsilon/d<1$. Then, $y_0 = P_{\B}(y) + \left(1-\frac{\epsilon}{d}\right)P_A^c(y)$. Note that if $i\in I$ then
\[\begin{aligned}
\text{sign}(y_0[i])&=\text{sign}\left(P_{\B}(y)[i] + \left(1-\frac{\epsilon}{d}\right)P_A^c(y)[i]\right) \\&= \text{sign}\left((|(Ax_0)[i]|-\left(1-\frac{\epsilon}{d}\right)|P_A^c(y)[i]|)\text{sign}(y[i])\right) \\&= \text{sign}(y[i]),
\end{aligned}\]
as from~\eqref{eq:sandwich}
\[\begin{aligned}
|(Ax_0)[i]|-|P_A^c(y)[i]|+\epsilon(|P_A^c(y)[i]|/d)&>\epsilon\left(|P_A^c(y)[i]/d - 1\right)  \geq 0.
\end{aligned}\]
Together with Lemma~\ref{lem:not_in_i}, we conclude that $\text{sign}(y_0[i])=\text{sign}(y[i])$ for all $i$.
Therefore, with~\eqref{eq:inB} we have $$P_{\B}(y_0)=P_{\B}(y) = P_A(y_0),$$ so $y_0$ corresponds to a solution. 
In addition, 
\begin{align*}
||y-y_0||_2 &\leq ||w||_2 + \frac{\epsilon}{d}||P_A^c(y)||_2 \\& < \epsilon + \frac{\epsilon}{d}||P_A^c(y)||_2;
\end{align*}
this completes the second part of the proof. 

If $\min_i|y[i]|\geq\epsilon$, we must have $I=\emptyset$, as if $i\in I$ then $$|y[i]| = |P_{\B}(y)[i]+w[i]+P_A^c(y)[i]|\leq  |w[i]| + |P_A^c(y)[i]|-|(Ax_0)[i]|\leq |w[i]|<\epsilon,$$ a contradiction. In that case we may set $\alpha=0$ in the above and conclude that $y_0 = P_{\B}(y) + P_A^c(y)$ corresponds to a solution and $||y-y_0||_2 = ||w||_2<\epsilon$.

\subsection{Proof of Lemma~\ref{lem:no_escape}}\label{sec:pf_no_escape}
For $\beta>\max_i\left|\frac{y[i]}{d[i]}\right|$, we have $P_{\B}(y-\beta d) = P_{\B}(-\beta d) = -P_{\B}(d)$. Then,
\[\begin{aligned} ||(y-\beta d) - P_AP_{\B}(y-\beta d)||_2^2 &= ||y-\beta d + P_AP_{\B}(d)||_2^2\\
&= \beta^2||d||_2^2 + ||y + P_AP_{\B}(d)||_2^2 - 2\beta\langle d, y + P_AP_{\B}(d)\rangle,\end{aligned}\]
where the second term is independent of $\beta$. 
Similarly, since $P_A$ is linear,
\[\begin{aligned} ||(y-\beta d) - P_A(y-\beta d)||_2^2 &= ||y-\beta d - P_A(y) + \beta P_A(d)||_2^2\\
&= \beta^2||(I-AA^{\dagger})d||_2^2 + ||y-P_A(y)||_2^2 - 2\beta \langle d, y-P_A(y)\rangle.\end{aligned}\]
In addition,
\[\begin{aligned} ||y-\beta d - P_{\B}(y-\beta d)||_2^2 &= ||y-\beta d + P_{\B}(d)||_2^2\\
&= \beta^2||d||_2^2 + ||y+P_{\B}(d)||_2^2 - 2\beta\langle d, y+P_{\B}(d)\rangle.\end{aligned}\]
Putting everything together:
\[\begin{aligned}
\fR(y-\beta d) &= \frac{1}{2}\beta^2 \|P_A(d)\|_2^2 - \beta\langle d, P_A(y + 2P_{\B}(d)) - P_{\B}(y)\rangle + c,
\end{aligned}\]
where $$c = ||y + P_AP_{\B}(d)||_2^2 - \frac{1}{2}\left(||y-P_A(y)||_2^2 + ||y+P_{\B}(d)||_2^2\right),$$ is independent of $\beta$.
If $P_A(d)\neq 0$ then $$\lim_{\beta\to\infty}\fR(y-\beta d)=\infty.$$ If $P_A(d) = 0$ then $$\fR(y-\beta d) = \beta\langle d, P_{\B}(y)\rangle + c,$$ so if $\langle d,P_{\B}(y)\rangle > 0$ we again have $\lim_{\beta\to\infty}\fR(y-\beta d)=\infty$.

\subsection{Proof of Corollary~\ref{cor:RRR_no_escape}}\label{sec:pf_RRR_no_escape}
For RRR, we have $$d=\nabla \fR(y) = P_A(y) + P_{\B}(y) - 2P_AP_{\B}(y),$$ and thus
\begin{equation} \label{eq:cor_proof}
P_A(d) = P_A(y) - P_AP_{\B}(y).
\end{equation}
Therefore, $P_A(d) = 0$ implies $\nabla \fR(y) = P_A^cP_{\B}(y)$ and hence $$\langle \nabla \fR(y), P_{\B}(y)\rangle = ||P_A^cP_{\B}(y)||_2^2 \geq 0.$$ If $\langle \nabla \fR(y), P_{\B}(y)\rangle = 0$ then $P_{\B}(y) = P_AP_{\B}(y) = P_A(y)$, where the second equality follows from~\eqref{eq:cor_proof}. Therefore, $y\in\col(A)\cap\B$ is a solution and $\nabla \fR(y)=0$.
So, $\langle \nabla \fR(y), P_{\B}(y)\rangle = ||P_A^cP_{\B}(y)||_2^2$ is zero only if $y$ is already a solution; otherwise, it is positive and thus satisfies the conditions of Lemma~\ref{lem:no_escape}. 

\bibliographystyle{plain}

\begin{thebibliography}{10}

\bibitem{balakrishnan2017statistical}
Sivaraman Balakrishnan, Martin~J Wainwright, and Bin Yu.
\newblock Statistical guarantees for the {EM} algorithm: From population to
  sample-based analysis.
\newblock {\em The Annals of Statistics}, 45(1):77--120, 2017.

\bibitem{balan2006signal}
Radu Balan, Pete Casazza, and Dan Edidin.
\newblock On signal reconstruction without phase.
\newblock {\em Applied and Computational Harmonic Analysis}, 20(3):345--356,
  2006.

\bibitem{Bandeira2014}
Afonso~S Bandeira, Jameson Cahill, Dustin~G Mixon, and Aaron~A Nelson.
\newblock Saving phase: Injectivity and stability for phase retrieval.
\newblock {\em Applied and Computational Harmonic Analysis}, 37(1):106--125,
  2014.

\bibitem{barnett2018geometry}
Alexander Barnett, Charles~L Epstein, Leslie Greengard, and Jeremy Magland.
\newblock Geometry of the phase retrieval problem.
\newblock {\em arXiv preprint arXiv:1808.10747}, 2018.

\bibitem{Bendory2017}
Tamir Bendory, Robert Beinert, and Yonina~C Eldar.
\newblock Fourier phase retrieval: Uniqueness and algorithms.
\newblock In {\em Compressed Sensing and its Applications}, pages 55--91.
  Springer, 2017.

\bibitem{bendory2017non}
Tamir Bendory, Yonina~C Eldar, and Nicolas Boumal.
\newblock Non-convex phase retrieval from {STFT} measurements.
\newblock {\em IEEE Transactions on Information Theory}, 64(1):467--484, 2017.

\bibitem{bian2015fourier}
Liheng Bian, Jinli Suo, Guoan Zheng, Kaikai Guo, Feng Chen, and Qionghai Dai.
\newblock Fourier ptychographic reconstruction using {W}irtinger flow
  optimization.
\newblock {\em Optics express}, 23(4):4856--4866, 2015.

\bibitem{borwein2017reflection}
Jonathan~M Borwein and Matthew~K Tam.
\newblock Reflection methods for inverse problems with applications to protein
  conformation determination.
\newblock In {\em Generalized Nash Equilibrium Problems, Bilevel Programming
  and MPEC}, pages 83--100. Springer, 2017.

\bibitem{bostan2018accelerated}
Emrah Bostan, Mahdi Soltanolkotabi, David Ren, and Laura Waller.
\newblock Accelerated {W}irtinger flow for multiplexed {F}ourier ptychographic
  microscopy.
\newblock In {\em 2018 25th IEEE International Conference on Image Processing
  (ICIP)}, pages 3823--3827. IEEE, 2018.

\bibitem{boumal2016nonconvex}
Nicolas Boumal.
\newblock Nonconvex phase synchronization.
\newblock {\em SIAM Journal on Optimization}, 26(4):2355--2377, 2016.

\bibitem{boumal2018deterministic}
Nicolas Boumal, Vladislav Voroninski, and Afonso~S Bandeira.
\newblock Deterministic guarantees for {B}urer-{M}onteiro factorizations of
  smooth semidefinite programs.
\newblock {\em Communications on Pure and Applied Mathematics}, 2018.

\bibitem{boyd2011distributed}
Stephen Boyd, Neal Parikh, Eric Chu, Borja Peleato, and Jonathan Eckstein.
\newblock Distributed optimization and statistical learning via the alternating
  direction method of multipliers.
\newblock {\em Foundations and Trends{\textregistered} in Machine learning},
  3(1):1--122, 2011.

\bibitem{cai2016optimal}
T~Tony Cai, Xiaodong Li, and Zongming Ma.
\newblock Optimal rates of convergence for noisy sparse phase retrieval via
  thresholded {W}irtinger flow.
\newblock {\em The Annals of Statistics}, 44(5):2221--2251, 2016.

\bibitem{candes2015phase_WF}
Emmanuel~J Candes, Xiaodong Li, and Mahdi Soltanolkotabi.
\newblock Phase retrieval via {W}irtinger flow: Theory and algorithms.
\newblock {\em IEEE Transactions on Information Theory}, 61(4):1985--2007,
  2015.

\bibitem{chen2017solving}
Yuxin Chen and Emmanuel~J. Candes.
\newblock Solving random quadratic systems of equations is nearly as easy as
  solving linear systems.
\newblock {\em Communications on Pure and Applied Mathematics}, 70(5):822--883,
  2017.

\bibitem{Chen2019}
Yuxin Chen, Yuejie Chi, Jianqing Fan, and Cong Ma.
\newblock Gradient descent with random initialization: fast global convergence
  for nonconvex phase retrieval.
\newblock {\em Mathematical Programming}, 176(1):5--37, Jul 2019.

\bibitem{Conca2015}
Aldo Conca, Dan Edidin, Milena Hering, and Cynthia Vinzant.
\newblock An algebraic characterization of injectivity in phase retrieval.
\newblock {\em Applied and Computational Harmonic Analysis}, 38(2):346--356,
  2015.

\bibitem{GS_counterexample}
Etienne Corman and Xiaoming Yuan.
\newblock A generalized proximal point algorithm and its convergence rate.
\newblock {\em SIAM Journal on Optimization}, 24(4):1614--1638, 2014.

\bibitem{douglas1956numerical}
Jim Douglas and Henry~H Rachford.
\newblock On the numerical solution of heat conduction problems in two and
  three space variables.
\newblock {\em Transactions of the American mathematical Society},
  82(2):421--439, 1956.

\bibitem{Eckstein1992}
Jonathan Eckstein and Dimitri~P. Bertsekas.
\newblock On the {D}ouglas---{R}achford splitting method and the proximal point
  algorithm for maximal monotone operators.
\newblock {\em Mathematical Programming}, 55(1):293--318, Apr 1992.

\bibitem{eldar2014phase}
Yonina~C Eldar and Shahar Mendelson.
\newblock Phase retrieval: Stability and recovery guarantees.
\newblock {\em Applied and Computational Harmonic Analysis}, 36(3):473--494,
  2014.

\bibitem{elser2003phase}
Veit Elser.
\newblock Phase retrieval by iterated projections.
\newblock {\em JOSA A}, 20(1):40--55, 2003.

\bibitem{Elser2017a}
Veit Elser.
\newblock Matrix product constraints by projection methods.
\newblock {\em Journal of Global Optimization}, 68(2):329--355, 2017.

\bibitem{Elser2018}
Veit Elser.
\newblock The complexity of bit retrieval.
\newblock {\em IEEE Transactions on Information Theory}, 64(1):412--428, 2018.

\bibitem{Elser2017}
Veit Elser, Ti-Yen Lan, and Tamir Bendory.
\newblock Benchmark problems for phase retrieval.
\newblock {\em SIAM Journal on Imaging Sciences}, 11(4):2429--2455, 2018.

\bibitem{elser2007searching}
Veit Elser, I~Rankenburg, and P~Thibault.
\newblock Searching with iterated maps.
\newblock {\em Proceedings of the National Academy of Sciences},
  104(2):418--423, 2007.

\bibitem{Fienup1982}
James~R Fienup.
\newblock Phase retrieval algorithms: a comparison.
\newblock {\em Applied optics}, 21(15):2758--2769, 1982.

\bibitem{gerchberg1972practical}
Ralph~W Gerchberg and Saxton W.O.
\newblock A practical algorithm for the determination of phase from image and
  diffraction plane pictures.
\newblock {\em Optik}, 35:237--246, 1972.

\bibitem{hesse2013nonconvex}
Robert Hesse and D~Russell Luke.
\newblock Nonconvex notions of regularity and convergence of fundamental
  algorithms for feasibility problems.
\newblock {\em SIAM Journal on Optimization}, 23(4):2397--2419, 2013.

\bibitem{lee2016gradient}
Jason~D Lee, Max Simchowitz, Michael~I Jordan, and Benjamin Recht.
\newblock Gradient descent only converges to minimizers.
\newblock In {\em Conference on learning theory}, pages 1246--1257, 2016.

\bibitem{li2016douglas}
Guoyin Li and Ting~Kei Pong.
\newblock Douglas--{R}achford splitting for nonconvex optimization with
  application to nonconvex feasibility problems.
\newblock {\em Mathematical programming}, 159(1-2):371--401, 2016.

\bibitem{Li2017a}
Ji~Li and Tie Zhou.
\newblock On relaxed averaged alternating reflections ({RAAR}) algorithm for
  phase retrieval with structured illumination.
\newblock {\em Inverse Problems}, 33(2):025012, 2017.

\bibitem{li2019rapid}
Xiaodong Li, Shuyang Ling, Thomas Strohmer, and Ke~Wei.
\newblock Rapid, robust, and reliable blind deconvolution via nonconvex
  optimization.
\newblock {\em Applied and computational harmonic analysis}, 47(3):893--934,
  2019.

\bibitem{lindstrom2018survey}
Scott~B Lindstrom and Brailey Sims.
\newblock Survey: Sixty years of {D}ouglas--{R}achford.
\newblock {\em arXiv preprint arXiv:1809.07181}, 2018.

\bibitem{Luke2005}
D~Russell Luke.
\newblock Relaxed averaged alternating reflections for diffraction imaging.
\newblock {\em Inverse problems}, 21(1):37, 2004.

\bibitem{Marchesini2007}
Stefano Marchesini.
\newblock Phase retrieval and saddle-point optimization.
\newblock {\em JOSA A}, 24(10):3289--3296, 2007.

\bibitem{netrapalli2013phase}
Praneeth Netrapalli, Prateek Jain, and Sujay Sanghavi.
\newblock Phase retrieval using alternating minimization.
\newblock In {\em Advances in Neural Information Processing Systems}, pages
  2796--2804, 2013.

\bibitem{boyd_prox_algs}
Neal Parikh and Stephen Boyd.
\newblock Proximal algorithms.
\newblock {\em Foundations and Trends® in Optimization}, 1(3):127--239, 2014.

\bibitem{pauwels2017fienup}
Edouard Jean~Robert Pauwels, Amir Beck, Yonina~C Eldar, and Shoham Sabach.
\newblock On {F}ienup methods for sparse phase retrieval.
\newblock {\em IEEE Transactions on Signal Processing}, 66(4):982--991, 2017.

\bibitem{phan2016linear}
Hung~M Phan.
\newblock Linear convergence of the {D}ouglas--{R}achford method for two closed
  sets.
\newblock {\em Optimization}, 65(2):369--385, 2016.

\bibitem{pinilla2019frequency}
Samuel Pinilla, Tamir Bendory, Yonina~C Eldar, and Henry Arguello.
\newblock Frequency-resolved optical gating recovery via smoothing gradient.
\newblock {\em IEEE Transactions on Signal Processing}, 67(23):6121--6132,
  2019.

\bibitem{shechtman2015phase}
Yoav Shechtman, Yonina~C Eldar, Oren Cohen, Henry~Nicholas Chapman, Jianwei
  Miao, and Mordechai Segev.
\newblock Phase retrieval with application to optical imaging: a contemporary
  overview.
\newblock {\em IEEE Signal Processing Magazine}, 32(3):87--109, 2015.

\bibitem{sun2016complete}
Ju~Sun, Qing Qu, and John Wright.
\newblock Complete dictionary recovery over the sphere ii: Recovery by
  {R}iemannian trust-region method.
\newblock {\em IEEE Transactions on Information Theory}, 63(2):885--914, 2016.

\bibitem{sun2018geometric}
Ju~Sun, Qing Qu, and John Wright.
\newblock A geometric analysis of phase retrieval.
\newblock {\em Foundations of Computational Mathematics}, 18(5):1131--1198,
  2018.

\bibitem{waldspurger2018phase}
Ir{\`e}ne Waldspurger.
\newblock Phase retrieval with random gaussian sensing vectors by alternating
  projections.
\newblock {\em IEEE Transactions on Information Theory}, 64(5):3301--3312,
  2018.

\bibitem{wang2017solving}
Gang Wang, Georgios~B Giannakis, and Yonina~C Eldar.
\newblock Solving systems of random quadratic equations via truncated amplitude
  flow.
\newblock {\em IEEE Transactions on Information Theory}, 64(2):773--794, 2017.

\bibitem{xu2018accelerated}
Rui Xu, Mahdi Soltanolkotabi, Justin~P Haldar, Walter Unglaub, Joshua Zusman,
  Anthony~FJ Levi, and Richard~M Leahy.
\newblock Accelerated {W}irtinger flow: A fast algorithm for ptychography.
\newblock {\em arXiv preprint arXiv:1806.05546}, 2018.

\bibitem{yeh2015experimental}
Li-Hao Yeh, Jonathan Dong, Jingshan Zhong, Lei Tian, Michael Chen, Gongguo
  Tang, Mahdi Soltanolkotabi, and Laura Waller.
\newblock Experimental robustness of {F}ourier ptychography phase retrieval
  algorithms.
\newblock {\em Optics express}, 23(26):33214--33240, 2015.

\bibitem{zhang2019phase}
Teng Zhang.
\newblock Phase retrieval using alternating minimization in a batch setting.
\newblock {\em Applied and Computational Harmonic Analysis}, 2019.

\end{thebibliography}

\end{document}